%% file: paper-submitted.tex
\documentclass{llncs}
\usepackage{boxedminipage}     
\usepackage{proof}
\usepackage{amssymb} 
\usepackage{amsmath} 
\usepackage{xspace}
\usepackage{times}

\newcommand{\citet}[1]{\cite{#1}}
\input{nominal}
\input{macros}
\begin{document}

\title{Regular Expression Subtyping for XML Query and Update
  Languages}

\author{James Cheney}
\institute{University of Edinburgh}
\maketitle

\begin{abstract}
\input{abstract.tex}
\end{abstract}

\section{Introduction}\labelSec{intro}

The Extensible Markup Language (XML) is a World Wide Web Consortium
(W3C) standard for tree-structured data.  Regular expression types for
XML~\cite{hosoya05toplas} have been studied extensively in XML
processing languages such as XDuce~\cite{hosoya03toit} and
CDuce~\cite{benzaken03icfp}, as well as projects to extend
general-purpose programming languages with XML features such as
Xtatic~\cite{DBLP:conf/planX/GapeyevGP06} and
OCamlDuce~\cite{frisch06icfp}.

Several other W3C standards, such as XQuery, address the use of XML as
a general format for representing data in databases.  Static
typechecking is important in XML database applications because type
information is useful for optimizing queries and avoiding expensive
run-time checks and revalidation.  The XQuery
standard~\cite{xquery-semantics-w3c-20070123} provides for structural
subtyping based on regular expression types.

However, XQuery's type system is imprecise in some situations
involving iteration (\verb|for|-expressions).  In particular, if the
variable $\$x$ has type\footnote{We use the notation for regular
  expression types from Hosoya, Vouillon and
  Pierce~\cite{hosoya05toplas} in preference to the more verbose
  XQuery or XML Schema syntaxes.}  $a[b[]^*,c[]^?]$, then the XQuery
expression
\begin{verbatim}
for $y in $x/* return $y
\end{verbatim}
has type $(b[]|c[])^*$ in XQuery, but in fact the result will always
match the regular expression type $b[]^*,c[]^?$.  The reason for this
inaccuracy is that XQuery's type system typechecks a \verb|for| loop
by converting the type of the body of the expression (here, $\$x/a$
with type $b[]^*,c[]^?$) to the ``factored'' form
$(\alpha_1|\ldots|\alpha_n)^q$, where $q$ is a quantifier such as $?$,
$+$, or $*$ and each $\alpha_i$ is an atomic type (i.e. a data type
such as $\kstring$ or single element type $a[\tau]$).

More precise type systems have been contemplated for XQuery-like
languages, including a precursor to XQuery designed by Fernandez,
Sim\'eon, and Wadler~\cite{fernandez01icdt}.  More recently, Colazzo
et al.~\cite{colazzo06jfp} have introduced a core XQuery language
called \muXQ, equipped with a regular expression-based type system
that provides more precise types for iterations using techniques
similar to those in~\cite{fernandez01icdt}.  In \muXQ, the above
expression can be assigned the more accurate type $b[]^*,c[]^?$.

Accurate typing for iteration constructs is especially important in
typechecking XML updates.  We are developing a statically-typed update
language called \Flux~\cite{cheney07planx} in which ideas from \muXQ
are essential for typechecking updates involving iteration.  Using
XQuery-style factoring for iteration in \Flux would make it impossible
to typecheck updates that modify data without modifying the overall
schema of the database---a very common case.  For example, using
XQuery-style factoring for iteration in \Flux, we would not be able to
verify statically that given a database of type
$a[b[\kstring]^*,c[]^?]$, an update that modifies the text inside
some of the $b$ elements produces an output that is still of type
$a[b[\kstring]^*,c[]^?]$, rather than
$a[(b[\kstring]|c[])^*]$.

One question left unresolved in previous work on both \muXQ and \Flux
is the relationship between declarative and algorithmic presentations
of the type system (in the terminology of ~\cite[Ch.
15--16]{pierce02types}).  Declarative derivations permit arbitrary
uses of the \emph{subsumption rule}:
\[
\infer{\wf{\Gamma}{e}{\tau'}}{\wf{\Gamma}{e}{\tau} & \tau \subty \tau'}
\]
whereas algorithmic derivations limit the use of this rule in order to
ensure that typechecking is syntax-directed and decidable.  The
declarative and algorithmic presentations of a system should agree.
If they do, then declarative typechecking is decidable; if they
disagree, then the algorithmic system is incomplete relative to the
high-level declarative system: it rejects programs that should
typecheck.

The XQuery standard circumvented this issue by directly defining
typechecking to be algorithmic.  In contrast, neither subsumption nor
subtyping were considered in \muXQ, in part because subtyping
interacts badly with \muXQ's ``path correctness'' analysis (as argued
by Colazzo et al.~\cite{colazzo06jfp}, Section 4.4).  Subsumption was
considered in our initial work on \Flux~\cite{cheney07planx}, but we
were initially unable to establish that declarative typechecking was
decidable, even in the absence of recursion in types, queries, or updates.

In this paper we consider declarative typechecking for \muXQ and \Flux
extended with recursive types, recursive functions, and recursive
update procedures.  To establish that typechecking remains decidable,
it suffices (following Pierce~\cite[Ch.  16]{pierce02types}) to define
an algorithmic typechecking judgment and prove its completeness; that
is, that declarative derivations can always be normalized to
algorithmic derivations.  For XQuery proper, this appears
straightforward because of the use of factoring when typechecking
iterations.  However, for \muXQ's more precise iteration type
discipline, completeness of algorithmic typechecking does not follow
by the ``obvious'' structural induction.  Instead, we must establish a
stronger property by considering the structure of regular expression
types.  We also extend these results to \Flux.

The structure of the rest of the paper is as follows.
\refSec{background} reviews regular expression types and subtyping.
\refSec{query-language} introduces the core language \muXQ, discusses
examples highlighting the difficulties involving subtyping in \muXQ,
and proves decidability of declarative typechecking.  We also review
the \Flux core update language in \refSec{update-language}, discuss
examples, and extend the proof of decidability of declarative
typechecking to \Flux.  Sections
\ref{sec:related-and-future-work}--\ref{sec:concl} sketch related and
future work and conclude.

\section{Background}\labelSec{background}

For the purposes of this paper, \emph{XML values} are trees built up
out of booleans $b \in \Bool = \{\ktrue,\kfalse\}$, strings $w \in
\Sigma^*$ over some alphabet $\Sigma$, and labels $l,m,n \in \Lab$,
according to the following syntax:
\[
\bar{v} ::= b \mid w \mid n[v] 
\qquad v ::= \bar{v},v \mid \emptyseq
\]
Values include \emph{tree values} $\bar{v}\in \Tree$ and
\emph{forest values} $v \in \Val$.  We write $v,v'$ for the
result of appending two forest values (considered as lists).

We consider a regular expression type system with structural
subtyping, similar to those considered in several transformation and
query languages for XML
\cite{hosoya05toplas,colazzo06jfp,fernandez01icdt}.  The syntax of
types and type environments is as follows.
\[\begin{array}{lrcl}
  \text{Atomic types} & \alpha &::=& \kbool \mid \kstring \mid n[\tau]
  \\
  \text{Sequence types} & \tau &::=& \alpha \mid \emptyseq \mid \tau |\tau' \mid \tau,\tau' \mid \tau^* \mid X\\
\text{Type definitions} &  \tau_0 &::=& \alpha \mid \emptyseq \mid \tau_0 |\tau_0' \mid \tau_0,\tau_0' \mid \tau_0^*\\
  \text{Type signatures} &E & ::= & \cdot \mid E,\typedecl{X}{\tau_0}
\end{array}\]
We call types of the form $\alpha \in \Atoms$ \emph{atomic} types (or
sometimes tree or singular types), and types $\tau \in \Types$ of all
other forms \emph{sequence types} (or sometimes forest or plural
types).  It should be obvious that a value of singular type must
always be a sequence of length one (that is, a tree); plural types may
have values of any length.  There exist plural types with only values
of length one, but which are not syntactically singular (for example
$\kint | \kbool$).  As usual, the $+$ and $?$ quantifiers can be
defined as follows: $\tau^+ = \tau,\tau^*$ and $\tau^? =
\tau|\emptyseq$.  We abbreviate $n[\emptyseq]$ as $n[]$.

Note that in contrast to Hosoya et al.~\cite{hosoya05toplas}, but
following Colazzo et al.~\cite{colazzo06jfp}, we include both Kleene
star and type variables.  In \cite{hosoya05toplas}, it was shown that
Kleene star can be translated away by introducing type variables and
definitions, modulo a syntactic restriction on top-level occurrences
of type variables.  In contrast, we allow Kleene star, but further
restrict type variables.  Recursive and mutually recursive
declarations are allowed, but type variables may not appear at the top
level of a type definition $\tau_0$: for example, $\typedecl{X}{
  nil[]|cons(a,X)}$ and
$\typedecl{Y}{\mathit{leaf}[]|\mathit{node}[X,X]}$ are allowed but
$\typedecl{X'}{ \emptyseq|a[],X}$ and $\typedecl{Y'}{b[] |
  Y',Y'}$ are not.  The equation for $X'$ defines the regular tree
language $a[]^*$, and would be permitted in XDuce, while that for $Y'$
defines a context-free tree language that is not regular.

An environment $E$ is well-formed if all type variables appearing in
definitions are themselves declared in $E$.  Given a well-formed
environment $E$, we write $E(X)$ for the definition of $X$.  A type
denotes the set of values $\SB{\tau}_E$, defined as follows.
\[\begin{array}{l}
  \begin{array}{rclcrclcrcl}
    \SB{\kstring}_E &=& \Sigma^*&&
    \SB{\kbool}_E &=& \Bool
&&
    \SB{\emptyseq}_E &=& \{\emptyseq\}\smallskip\\
    \SB{n[\tau]}_E &=& \{n[v] \mid v \in \SB{\tau}_E\}
&&\SB{X}_E &=& \SB{E(X)}
&&    \SB{\tau|\tau'}_E &=& \SB{\tau}_E \cup \SB{\tau'}_E
  \end{array}
  \smallskip\\
  \begin{array}{rcl}
    \SB{\tau,\tau'}_E &=& \{v, v' \mid v \in \SB{\tau}_E,v' \in \SB{\tau'}_E\}
    \smallskip\\
    \SB{\tau^*}_E &=& \{\emptyseq\} \cup \{v_1, \ldots, v_n \mid v_1 \in \SB{\tau}_E,\ldots,v_n \in \SB{\tau}_E\}\\
  \end{array}
\end{array}
\]
Formally, $\SB{\tau}_E$ must be defined by a least fixed point
construction which we take for granted.  Henceforth, we treat $E$ as
fixed and define $\SB{\tau} = \SB{\tau}_{E}$.

In addition, we define a binary \emph{subtyping} relation on types.  A
type $\tau_1$ is a subtype of $\tau_2$ ($\tau_1 \subty \tau_2$), by
definition, if $\SB{\tau_1} \subseteq \SB{\tau_2}$.  Our types can be
translated to XDuce types, so subtyping reduces to XDuce subtyping;
although this problem is EXPTIME-complete in general, the algorithm of
\citet{hosoya05toplas} is well-behaved in practice.  Therefore, we
shall not give explicit inference rules for checking or deciding
subtyping, but treat it as a ``black box''.

\section{Query language}\labelSec{query-language}

We review an XQuery-like core language based on
\muXQ~\citet{colazzo06jfp}.  In \muXQ, we distinguish between
\emph{tree variables} $\bar{x} \in \TVar$, introduced by $\kfor$, and
\emph{forest variables}, $x \in \Var$, introduced by $\klet$.  We
write $\hat{x} \in \Var \cup \TVar$ for an arbitrary variable.  The
other syntactic classes of our variant of \muXQ include booleans,
strings, and labels introduced above, function names $F \in \FSym$,
expressions $e \in \Expr$, and programs $p \in \Prog$; the abstract
syntax of expressions and programs is defined as follows:
\begin{eqnarray*}
  e &::=& \emptyseq \mid e,e' \mid n[e] \mid w \mid x \mid \letin{x=e}{e'}
  \mid F(e_1,\ldots,e_n) \\
  &\mid & b\mid \ifthenelse{c}{e}{e'}\mid \bar{x} \mid \bar{x}/\kchild \mid e::n  \mid \forreturn{\bar{x} \in e}{e'} \\
  p &::=&  \query{e}{\tau} \mid \declarefunction{F(x_1{:}\tau_1,\ldots,x_n{:}\tau_n):\tau}{e}{p} 
\end{eqnarray*}
The distinguished variables $\bar{x}$ in $\forreturn{\bar{x} \in
  e}{e'(\bar{x})}$ and $x$ in $\letin{x=e}{e'(x)}$ are bound in
$e'(x)$.  Here and elsewhere, we employ common conventions such as
considering expressions containing bound variables equivalent up to
$\alpha$-renaming and employing a richer concrete syntax including
parentheses.

To simplify the presentation, we split \muXQ's projection operation
$\bar{x}/\kchild::l$ into two expressions: child projection
($\bar{x}/\kchild$) which returns the children of $\bar{x}$, and node
name filtering ($e::n$) which evaluates $e$ to an arbitrary sequence
and selects the nodes labeled $n$.  Thus, the ordinary child axis
expression $\bar{x}/\kchild::n$ is syntactic sugar for
$(\bar{x}/\kchild)::n$ and the ``wildcard'' child axis is definable as
$\bar{x}/\kchild::* = \bar{x}/\kchild$.  Built-in operations such as
string equality may be provided as additional functions $F$.

Colazzo et al.~\cite{colazzo06jfp} provided a denotational semantics
of \muXQ queries with the descendant axis but without recursive
functions.  This semantics is sound with respect to the typing rules
in the next section and can be extended to handle recursive functions
using operational techniques (as in the XQuery standard).  However, we
omit the semantics since it is not needed in the rest of the paper.

\subsection{Type system}\labelSec{query-types}

\begin{figure}[tb]
  \fbox{$\wf{\Gamma}{e}{\tau}$}
  \[
  \begin{array}{c}
    \infer{\wf{\Gamma}{\bar{x}}{\alpha}}{\bar{x}{:}\alpha \in \Gamma}
    \quad
    \infer{\wf{\Gamma}{x}{\tau}}{x{:}\tau \in \Gamma}
    \quad
    \infer{\wf{\Gamma}{w}{\kstring}}{}
    \quad
    \infer{\wf{\Gamma}{b}{\kbool}}{b \in \Bool}
    \smallskip\\
    \infer{\wf{\Gamma}{\emptyseq}{\emptyseq}}{}
    \quad
    \infer{\wf{\Gamma}{n[e]}{n[\tau]}}{\wf{\Gamma}{e}{\tau}}
    \quad
    \infer{\wf{\Gamma}{e,e'}{\tau,\tau'}}{\wf{\Gamma}{e}{\tau} & \wf{\Gamma}{e'}{\tau'}}
    \quad
    \infer{\wf{\Gamma}{\letin{x=e_1}{e_2}}{\tau_2}}
    {\wf{\Gamma}{e_1}{\tau_1} & \wf{\Gamma,x{:}\tau_1}{e_2}{\tau_2}}
    \smallskip\\
    \infer{\wf{\Gamma}{\ifthenelse{c}{e_1}{e_2}}{\tau_1|\tau_2}}
    {\wf{\Gamma}{c}{\kbool} & 
      \wf{\Gamma}{e_1}{\tau_1} &
      \wf{\Gamma}{e_2}{\tau_2} 
    }
    \quad
    \infer{\wf{\Gamma}{\bar{x}/\kchild}{\tau}}{\bar{x}{:}n[\tau] \in \Gamma}
    \quad
    \infer{\wf{\Gamma}{e::n}{\tau'}}{\wf{\Gamma}{e}{\tau} 
      & \tylab{\tau}{n}{\tau'}}
    \smallskip\\
    \infer{\wf{\Gamma}{\forreturn{\bar{x} \in e_1}{e_2}}{\tau_2}}
    {\wf{\Gamma}{e_1}{\tau_1} 
      & 
      \wfin{\Gamma}{x}{\tau_1}{e_2}{\tau_2}}
    \quad
    \infer{\wf{\Gamma}{F(\vec{e})}{\tau_0}}{\funcdecl{F(\vec{\tau})}{\tau_0} \in \Delta & \wf{\Gamma}{e_i}{\tau_i}}
    \quad
    \infer{\wf{\Gamma}{e}{\tau'}}
    {\wf{\Gamma}{e}{\tau} & \tau \subty \tau'}
  \end{array}
  \]
  \fbox{$\wfprog{\Gamma}{p}$}
  \[
  \begin{array}{c}
    \infer{\wfprog{\Gamma}{e:\tau}}{\wf{\Gamma}{e}{\tau}}
\quad
    \infer{\wfprog{\Gamma}{\declarefunction{\funcdecl{F(\vec{\tau})}{\tau_0}}{e}{p}}}{
      \funcdecl{\text{$F$ not declared in $p$} & 
          F(\vec{\tau})}{\tau_0} \in \Delta&
      \wf{\Gamma,\vec{x}~{:}~\vec{\tau}}{e}{\tau_0} & 
      \wfprog{\Gamma}{p}}
  \end{array}\]
  \caption{Query and program
    well-formedness rules}\labelFig{query-program-wf}
\begin{minipage}[t]{5cm}
\fbox{$\tylab{\tau}{n}{\tau'}$}
  \[\begin{array}{c}
     \infer{\tylab{n[\tau]}{n}{n[\tau]}}{}
    \smallskip\\
   \infer{\tylab{X}{n}{\tau}}{\tylab{E(X)}{n}{\tau}}
    \quad
    \infer{\tylab{\alpha}{n}{ \emptyseq}}{\alpha \neq n[\tau]}
    \smallskip\\
    \infer{\tylab{\emptyseq}{n}{\emptyseq}}{}
    \quad
    \infer{\tylab{\tau_1^*}{n}{\tau_2^*}}
    {\tylab{\tau_1}{n}{\tau_2}}
    \smallskip\\
    \infer{\tylab{\tau_1,\tau_2}{n}{\tau_1',\tau_2'}}
    {\tylab{\tau_1}{n}{\tau_1'}
      &
      \tylab{\tau_2}{n}{\tau_2'}}
    \smallskip\\ 
    \infer{\tylab{\tau_1|\tau_2}{n}{\tau_1'|\tau_2'}}
    {\tylab{\tau_1}{n}{\tau_1'}
      &
      \tylab{\tau_2}{n}{\tau_2'}}
  \end{array}
  \]
\end{minipage}
\begin{minipage}[t]{6cm}
  \fbox{$\wfin{\Gamma}{x}{\tau}{e}{\tau'}$}
  \[\begin{array}{c}
    \infer{\wfin{\Gamma}{x}{\emptyseq}{e}{\emptyseq}}{}
    \quad \infer{\wfin{\Gamma}{x}{X}{e}{\tau}}{\wfin{\Gamma}{x}{E(X)}{e}{\tau}}
    \smallskip\\
    \infer{\wfin{\Gamma}{x}{\alpha}{e}{\tau}}{\wf{\Gamma,\bar{x}{:}\alpha}{e}{\tau}}
    \quad
    \infer{\wfin{\Gamma}{x}{\tau_1^*}{e}{\tau_2^*}}{\wfin{\Gamma}{x}{\tau_1}{e}{\tau_2}}
    \smallskip\\
    \infer{\wfin{\Gamma}{x}{\tau_1,\tau_2}{e}{\tau_1',\tau_2'}}
    {\wfin{\Gamma}{x}{\tau_1}{e}{\tau_1'} & 
      \wfin{\Gamma}{x}{\tau_2}{e}{\tau_2'}}
    \smallskip\\
    \infer{\wfin{\Gamma}{x}{\tau_1|\tau_2}{e}{\tau_1'|\tau_2'}}
    {\wfin{\Gamma}{x}{\tau_1}{e}{\tau_1'} & 
      \wfin{\Gamma}{x}{\tau_2}{e}{\tau_2'}}
  \end{array}\]
\end{minipage}
\caption{Auxiliary judgments}\labelFig{aux-wf}
\end{figure}

Our type system for queries is essentially that introduced for \muXQ
by \cite{colazzo06jfp}, excluding the path correctness component.  We
consider typing environments $\Gamma$ and global declaration
environments $\Delta$, defined as follows:
\[
\Gamma ::= \cdot \mid \Gamma,x{:}\tau \mid \Gamma,\bar{x} {:} \alpha  \qquad \Delta ::= \cdot \mid \Delta,\funcdecl{F(\vec{\tau})}{\tau_0}\]
Note that in $\Gamma$, tree variables may only be bound to atomic
types.  As usual, we assume that variables in type environments are
distinct; this convention implicitly constrains all inference rules.
We also write $\Gamma \subty \Gamma'$ to indicate that $\dom(\Gamma)
=\dom(\Gamma')$ and $\Gamma'(\hat{x}) \subty \Gamma(\hat{x})$ for all
$\hat{x} \in \dom(\Gamma)$.

The main typing judgment for queries is $\wf{\Gamma}{e}{\tau}$; we
also define a program well-formedness judgment $\wfprog{\Gamma}{p}$
which typechecks the bodies of functions.  Following
\citet{colazzo06jfp}, there are two auxiliary judgments,
$\wfin{\Gamma}{x}{\tau}{s}{\tau'}$, used for typechecking
$\kfor$-expressions, and $\tylab{\tau}{n}{\tau'}$, used for
typechecking label matching expressions $e::n$.  The rules for these
judgments are shown in Figures \ref{fig:query-program-wf} and
\ref{fig:aux-wf}.

We consider the typing rules to be implicitly parameterized by a fixed
global declaration environment $\Delta$.  Functions in XQuery have
global scope so we assume that the declarations for all the functions
declared in the program have already been added to $\Delta$ by a
preprocessing pass.  Additional declarations for built-in functions
might be included in $\Delta$ as well.  

The rules involving type variables in \refFig{aux-wf} look up the
variable's definition in $E$.  These judgments only inspect the
top-level of a type; they do not inspect the contents of element types
$n[\tau]$.  Since type definitions $\tau_0$ have no top-level type
variables, both judgments are terminating.  (This was argued
in detail by Colazzo et al.~\cite[Lem.  4.6]{colazzo06jfp}.)

\subsection{Examples}\labelSec{query-examples}

We first revisit the example in the introduction in order to
illustrate the operation of the rules.  Recall that $\bar{x}/*$ is
translated to $\bar{x}/\kchild$ in our core language.
\[\small\infer{\wf{\bar{x}{:}a[b[]^*,c[]^?]}{\forreturn{\bar{y} \in \bar{x}/\kchild}{\bar{y}}}{b[]^*,c[]^?}}{
  \infer{\wf{\bar{x}{:}a[b[]^*,c[]^?]}{\bar{x}/\kchild}{b[]^*,c[]^?}}{}
  & 
  \deduce{\wfin{\bar{x}{:}a[b[]^*,c[]^?]}{y}{b[]^*,c[]^?}{\bar{y}}{b[]^*,c[]^?}}{\DD}
}
\]
where the subderivation $\DD$ is
\[\small
\DD = \left.
\begin{array}{c}
\infer{\wfin{\bar{x}{:}a[b[]^*,c[]^?]}{y}{b[]^*,c[]^?}{\bar{y}}{b[]^*,c[]^?}}{
    \infer{\wfin{\bar{x}{:}a[b[]^*,c[]^?]}{y}{b[]^*}{\bar{y}}{b[]^*}}{
      \infer{\wfin{\bar{x}{:}a[b[]^*,c[]^?]}{y}{b[]}{\bar{y}}{b[]}}{
        \hyp{\wf{\bar{x}{:}a[b[]^*,c[]^?],\bar{y}{:}b[]}{\bar{y}}{b[]}}
      }}
    &
    \infer{\wfin{\bar{x}{:}a[b[]^*,c[]^?]}{y}{c[]^?}{\bar{y}}{c[]^?}}{
      \infer{\wfin{\bar{x}{:}a[b[]^*,c[]^?]}{y}{c[]}{\bar{y}}{c[]}}{
        \hyp{\wf{\bar{x}{:}a[b[]^*,c[]],\bar{y}{:}c[]}{\bar{y}}{c[]}}
      }
    }
  }
\end{array}\right.
\]
Note that this derivation does not use subsumption anywhere.  Suppose
we wished to show that the expression has type $b[]^*,(c[]^?|d[]^*)$, a
supertype of the above type.  There are several ways to do this:
first, we can simply use subsumption at the end of the derivation.
Alternatively, we could have used subsumption in one of the
subderivations such as
$\wf{\bar{x}{:}a[b[]^*,c[]^?],\bar{y}{:}c[]^?}{\bar{y}}{c[]^?}$, to conclude,
for example, that
$\wf{\bar{x}{:}a[b[]^*,c[]^?],\bar{y}{:}c[]^?}{\bar{y}}{c[]^?|d[]^*}$.  This
is valid since $c[]^? \subty c[]^?|d[]^*$.

Suppose, instead, that we actually wanted to show that the above
expression has type $(b[d[]^*]|c[]^?)^*$, also a supertype of the
derived type.  There are again several ways of doing this.  Besides
using subsumption at the end of the derivation, we might have used it
on $\wf{\bar{x}{:}a[b[]^*,c[]^?]}{\bar{x}/\kchild}{b[]^*,c[]^?}$ to
obtain $\wf{\bar{x}{:}a[b[]^*,c[]^?]}{\bar{x}/\kchild}{(b[d[]^*]|c[]^?)^*}$.
To complete the derivation, we would then need to replace derivation
$\DD$ with $\DD'$:
\[\small
\DD' = \left.\begin{array}{c}
 \infer{\wfin{\bar{x}{:}a[b[]^*,c[]^?]}{y}{(b[d[]^*]|c[]^?)^*}{\bar{y}}{(b[d[]^*]|c[]^?)^*}}{
    \infer{\wfin{\bar{x}{:}a[b[]^*,c[]^?]}{y}{b[d[]^*]|c[]^?}{\bar{y}}{b[d[]^*]|c[]^?}}{
    \infer{\wfin{\bar{x}{:}a[b[]^*,c[]^?]}{y}{b[d[]^*]}{\bar{y}}{b[d[]^*]}}{
      \hyp{\wf{\bar{x}{:}a[b[]^*,c[]^?],\bar{y}{:}b[d[]^*]}{\bar{y}}{b[d[]^*]}}
    }&
    \infer{\wfin{\bar{x}{:}a[b[]^*,c[]^?]}{y}{c[]^?}{\bar{y}}{c[]^?}}{
      \infer{\wfin{\bar{x}{:}a[b[]^*,c[]^?]}{y}{c[]}{\bar{y}}{c[]}}{
        \hyp{\wf{\bar{x}{:}a[b[]^*,c[]^?],\bar{y}{:}c[]}{\bar{y}}{c[]}}
      }
    }
  }
}
\end{array}\right.\]
Not only does $\DD'$ have different structure than $\DD$, but
it also requires subderivations that were not syntactically present in
$\DD$.  

The above example illustrates why eliminating uses of subsumption is
tricky.  If subsumption is used to weaken the type of the first
argument of a $\kfor$-expression according to $\tau_1' \subty \tau_1$,
then we need to know that we can transform the corresponding
derivation $\DD$ of $\wfin{\Gamma}{x}{\tau_1}{e}{\tau_2}$ to a
derivation of $\DD'$ of $\wfin{\Gamma}{x}{\tau_1'}{e}{\tau_2'}$ for
some $\tau_2' \subty \tau_2$.  But as illustrated above, the
derivations $\DD$ and $\DD'$ may bear little resemblance to one
another.

Now we consider a typechecking a recursive query.  Suppose we have
$\typedecl{\Tree}{\tree[\leaf[\kstring]|\node[\Tree^*]]}$ and function definition
\[\begin{array}{l}
\declarefunction{\funcdecl{\leaves(x:\Tree)}{\leaf[\kstring]^*}}{\\
\quad x/\leaf,\forreturn{\bar{z}\in x/\node/*}{\leaves(\bar{z})}\\
}
\end{array}
\]
This uses a construct $e/n$ that is not in core \muXQ, but we can
expand $e/n$ to $\forreturn{\bar{y}\in e}{\bar{y}/\kchild::n}$; thus,
we can derive a rule
\[
\infer{\wf{\Gamma}{e/n}{\tau'}}{\wf{\Gamma}{e}{l[\tau]} &
  \tylab{\tau}{n}{\tau'}}
\Longleftrightarrow
\infer{\wf{\Gamma}{\forreturn{\bar{y}\in e}{\bar{y}/\kchild::n}}{\tau'}}{
  \wf{\Gamma}{e}{l[\tau]} & 
  \infer{\wfin{\Gamma}{y}{l[\tau]}{\bar{y}/\kchild::n}{\tau'}}{
    \infer{\wf{\Gamma,\bar{y}{:}l[\tau]}{\bar{y}/\kchild::n}{\tau'}}{
      \infer{\wf{\Gamma,\bar{y}{:}l[\tau]}{\bar{y}/\kchild}{\tau}}{
      }
      & \tylab{\tau}{n}{\tau'}
    }
  }
}
\]
Using this derived rule and the fact that $x : \Tree$ and the
definition of $\Tree$, we can see that $x/\leaf:\leaf[\kstring]$ and
$x/\node : \node[\Tree^*]]$, and so $x/\node/* :
\tree[\leaf[\kstring]|\node[\Tree^*]]^*$.  So each iteration of the
$\kfor$-loop can be typechecked with
$\bar{z}:\tree[\leaf[\kstring]|\node[\Tree^*]]$.  To check the
function call $\leaves(\bar{z})$, we need subsumption to see that
$\tree[\leaf[\kstring]|\node[\Tree^*]]^* \subty \Tree$.  It follows
that that $\leaves(\bar{z}) : \leaf[\kstring]^*$, so the $\kfor$-loop
has type $(\leaf[\kstring]^*)^*$.  Again using subsumption, we can
conclude that 
\[x/\leaf,\leaves(x/\node/*) :
\leaf[\kstring],(\leaf[\kstring]^*)^* \subty \leaf[\kstring]^*\;.\]
Notice that although we could have used subsumption in several more
places, we really \emph{needed} it in only two places: when
typechecking a function call, and when checking the result of a
function against its declared type.
\subsection{Decidability}\labelSec{query-decidability}

The standard approach (see e.g. Pierce~\cite[Ch. 16]{pierce02types})
to deciding declarative typechecking is to define algorithmic
judgments that are syntax-directed and decidable, and then show that
the algorithmic system is complete relative to the declarative system.

\begin{definition}[Algorithmic derivations]
  The algorithmic typechecking judgments $\wfalg{\Gamma}{e}{\tau}$ and 
  $\wfinalg{\Gamma}{x}{\tau_0}{e}{\tau}$ are defined by taking the rules
  of Figures~\ref{fig:query-program-wf} and~\ref{fig:aux-wf}, removing
  the subsumption rule, and replacing the function application rule
  with
  \[
  \infer{\wfalg{\Gamma}{F(\vec{e})}{\tau}}{
    \funcdecl{F(\vec{\tau})}{\tau} \in \Gamma 
    & \wfalg{\Gamma}{e_i}{\tau_i'} & \tau_i' \subty \tau_i}
  \]
\end{definition}

It is straightforward to show that algorithmic derivability is
decidable and sound with respect to the declarative system:
\begin{lemma}[Decidability]\labelLem{decidability}
  For any $\bar{x},e,n$, there exist computable partial functions
  $f_n$, $g_e$, $h_{\bar{x},y}$ such that for any $\Gamma,\tau_0$, we have:
  \begin{enumerate}
  \item $f_n(\tau_0)$ is the unique $\tau$ such that
    $\tylab{\tau_0}{n}{\tau}$.
  \item $g_x(\Gamma)$ is the unique $\tau$ such that
    $\wfalg{\Gamma}{e}{\tau}$, when it exists.
  \item $h_{\bar{x},e}(\Gamma,\tau_0)$ is the unique $\tau$ such that
    $\wfinalg{\Gamma}{x}{\tau_0}{e}{\tau}$, when it
    exists.
  \end{enumerate}
\end{lemma}
\begin{theorem}[Algorithmic Soundness]\labelThm{soundness}
  (1) If $\wfalg{\Gamma}{e}{\tau}$ is derivable then $\wf{\Gamma}{e}{\tau}$ is derivable.  (2)
  If $\wfinalg{\Gamma}{x}{\tau_0}{e}{\tau}$ is derivable then
  $\wfin{\Gamma}{x}{\tau_0}{e}{\tau}$ is derivable.
\end{theorem}

The corresponding completeness property (the main result of this
section) is:
\begin{theorem}[Algorithmic Completeness]\labelThm{alg-completeness}
(1) If $\wf{\Gamma}{e}{\tau}$ then there exists $\tau' \subty
    \tau$ such that $\wfalg{\Gamma}{e}{\tau'}$.
(2) If $\wfin{\Gamma}{x}{\tau_1}{e}{\tau_2}$ then there exists
    $\tau_2' \subty \tau_2$ such that
    $\wfinalg{\Gamma}{x}{\tau_1}{e}{\tau_2'}$.
\end{theorem}
Given a decidable subtyping relation $\subty$, a typical proof of
completeness involves showing by induction that occurrences of the
subsumption rule can be ``permuted'' downwards in the proof past other
rules, except for function applications.  Completeness for \muXQ
requires strengthening this induction hypothesis.  To see why, recall
the following rules:
\[\small
\infer{\wf{\Gamma}{\letin{x=e_1}{e_2}}{\tau_2}}{
  \deduce{\wf{\Gamma}{e_1}{\tau_1}}{*} & 
  \wf{\Gamma,x{:}\tau_1}{e_2}{\tau_2}
}
\quad
\infer{\wf{\Gamma}{\forreturn{\bar{x} \in e_1}{e_2}}{\tau_2}}{
  \deduce{\wf{\Gamma}{e_1}{\tau_1} }{*}
  & 
  \wfin{\Gamma}{x}{\tau_1}{e_2}{\tau_2}
}
\quad
\infer{\wf{\Gamma}{e::n}{\tau'}}{
  \deduce{\wf{\Gamma}{e}{\tau} }{*}
  & \tylab{\tau}{n}{\tau'}
}
\]
If the subderivation labeled $*$ in the above rules follows by
subsumption, however, we cannot do anything to get rid of the
subsumption rule using the induction hypotheses provided by
\refThm{alg-completeness}.  Instead we need an additional lemma that
ensures that the judgments are all \emph{downward monotonic}.
Downward monotonicity means, informally, that if make the ``input''
types in a derivable judgment smaller, then the judgment
remains derivable with a smaller ``output'' type.

\begin{lemma}[Downward monotonicity]\labelLem{downward-monotonicity}
  ~
  \begin{enumerate}
  \item If $\tylab{\tau_1}{n}{\tau_2}$ and $\tau_1' \subty \tau_1$
    then $\tylab{\tau_1'}{n}{\tau_2'}$ for some $\tau_2' \subty
    \tau_2$
  \item If $\wfalg{\Gamma}{e}{\tau}$ and $\Gamma' \subty \Gamma$ then
    $\wfalg{\Gamma'}{e}{\tau'}$ for some $\tau' \subty \tau$.
  \item If $\wfinalg{\Gamma}{x}{\tau_1}{e}{\tau_2}$ and $\Gamma'
    \subty \Gamma$ and $\tau_1' \subty \tau_1$ then
    $\wfinalg{\Gamma'}{x}{\tau_1'}{e}{\tau_2'}$ for some $\tau_2'
    \subty \tau_2$.
  \end{enumerate}
\end{lemma}

The downward monotonicity lemma is \emph{almost} easy to prove by
direct structural induction (simultaneously on all judgments).  The
cases for (2) involving expression-directed typechecking are all
straightforward inductive steps; however, for the cases involving
type-directed judgments, the induction steps do not go through.  The
difficulty is illustrated by the following cases.  For derivations of
the form
\[
\infer{\tylab{\tau_1^*}{n}{\tau_2^*}}{\tylab{\tau_1}{n}{\tau_2}}
\qquad 
\infer{\wfin{\Gamma}{x}{\tau_1^*}{e}{\tau_2^*}}{\wfin{\Gamma}{x}{\tau_1}{e}{\tau_2}}
\]
we are stuck: knowing that $\tau_1' \subty \tau_1^*$ does not
necessarily tell us anything about a subtyping relationship between
$\tau_1'$ and $\tau_1$.  For example, if $\tau_1' = aa$ and $\tau_1 =
a$, then we have $aa \subty a^*$ but not $aa \subty a$.  Instead, we
need to proceed by an analysis of regular expression types and
subtyping.

We briefly sketch the argument, which involves an excursion into the
theory of regular languages over partially ordered alphabets.  Here,
the ``alphabet'' is the set of atomic types and the regular sets are
the sets of sequences of atomic types that are subtypes of a type
$\tau$.  The \emph{homomorphic extension} of a (possibly partial)
function $h : \Atoms \pto \Types$ on atomic types is defined as
\[
\begin{array}{rclcrclcrcl}
\hat{h}(\emptyseq) &=& \emptyseq&&
\hat{h}(\alpha) &=&  h(\alpha)&& \hat{h}(\tau^*) &=&  \hat{h}(\tau)^* \\
\hat{h}(\tau_1,\tau_2) &=&  \hat{h}(\tau_1),\hat{h}(\tau_2) &&
\hat{h}(\tau_1|\tau_2) &=& \hat{h}(\tau_1)|\hat{h}(\tau_2)&&
\hat{h}(X) &=& \hat{h}(E(X))
\end{array}\]
(Note again that this definition is well-founded, since type variables
cannot be expanded indefinitely.)  If $h$ is partial, then $\hat{h}$
is defined only on types whose atoms are in $\dom(h)$.  We can then
show the following general property of partial homomorphic extensions:
\begin{lemma}
  If $h:\Atoms\pto \Types$ is downward monotonic, then its homomorphic
  extension $\hat{h}: \Types \pto \Types$ is downward monotonic.
\end{lemma}

It then suffices to show that $f_n$ and $h_{\bar{x},e}$ are partial
homomorphic extensions of downward monotone functions on atomic types;
for $f_n$, the required function is simple and obviously monotone, and
for $h_{\bar{x},e}(\Gamma,-)$, the required generating function is
$g_e(\Gamma,x{:}(-))$.  Thus, we need to show that $g_e$ and
$h_{\bar{x},e}$ are downward monotonic and that
$h_{\bar{x},e}(\Gamma,-)$ is the partial homomorphic extension of
$g_e(\Gamma,x{:}(-))$ simultaneously by mutual induction.  This,
finally, is a straightforward induction over derivations.  More
detailed proofs are included in the appendix.

\section{Update language}\labelSec{update-language}

We now introduce the core \Flux update language, which extends the
syntax of queries with statements $s \in \Stmt$, procedure names $P
\in \PSym$, tests $\phi \in \Test$, directions $d \in \Dir$, and
two new cases for programs:
\begin{eqnarray*}
  s &::=& \kskip \mid s;s' \mid \ifthenelse{e}{s}{s'} \mid\letin{x=e}{s} \mid P(\vec{e})\\
  &\mid & \kinsert~e \mid \kdelete \mid \krename~n \mid \snapshot{x}{s}  \mid \phi?s \mid d[s]\\
  \phi &::=& n \mid * \mid \kbool \mid \kstring \qquad d ::= \kleft\mid \kright \mid \kchildren \mid \kiter\\
  p &::=& \cdots \mid \update{s}{\tau}{ \tau'} \mid \declareprocedure{\procdecl{P(\vec{x}:\vec{\tau})}{\tau}{\tau'}}{s}{p}
\end{eqnarray*}
Updates include standard programming constructs such as the no-op
$\kskip$, sequential composition, conditionals, and $\klet$-binding.
The basic update operations include insertion $\kinsert~e$, which
inserts a value into an empty part of the database; deletion
$\kdelete$, which deletes part of the database; and $\krename~n$,
which renames a part of the database provided it is a single tree.
The ``snapshot'' operation $\snapshot{x}{s}$ binds $x$ to part of the
database and then applies an update $s$, which may refer to $x$.  Note
that the snapshot operation is the only way to read from the
current database state.

Updates also include \emph{tests} $\phi?s$ which test the top-level
type of a singular value and conditionally perform an update,
otherwise do nothing.  The node label test $n?s$ checks whether the
tree is of type $n[\tau]$, and if so executes $s$; the wildcard test
$*?s$ checks that the value is a tree.  Similarly, $\kbool?s$ and
$\kstring?s$ test whether a value is a boolean or string.
The $?$ operator binds tightly; for example, $\phi?s;s' =
(\phi?s);s'$.

Finally, updates include \emph{navigation} operators that change the
selected part of the tree, and perform an update on the sub-selection.
The $\kleft$ and $\kright$ operators perform an update (typically, an
$\kinsert$) on the empty sequence located to the left or right of a
value.  The $\kchildren$ operator applies an update to the child list
of a tree value.  The $\kiter$ operator applies an update to each tree
value in a forest.

We distinguish between \emph{singular} (unary) updates which apply
only when the context is a tree value and \emph{plural} (multi-ary)
updates which apply to a sequence.  Tests $\phi?s$ are always
singular.  The $\kchildren$ operator applies a plural update to all of
the children of a single node; the $\kiter$ operator applies a
singular update to all of the elements of a sequence.  Other updates
can be either singular or plural in different situations.  Our type
system tracks multiplicity as well as input and output types in order
to ensure that updates are well-behaved.

\Flux updates operate on a part of the database that is ``in focus'',
which helps ensure that updates are deterministic and relatively easy
to typecheck.  Only the navigation operations
$\kleft,\kright,\kchildren,\kiter$ can change the focus.  We lack
space to formalize the semantics of updates in the main body of the
paper; the semantics of updates is essentially the same as in
\cite{cheney07planx} except for the addition of procedures.

\subsection{Type system}\labelSec{update-types}
\begin{figure}[tb]
  \fbox{$\wfupd{\Gamma}{a}{\tau}{s}{\tau'}$}
  \[\begin{array}{c}
    \infer{\wfupd{\Gamma}{a}{\tau}{\kskip}{\tau}}{}
    \quad
    \infer{\wfupd{\Gamma}{a}{\tau}{s;s'}{ \tau''}}
    {\wfupd{\Gamma}{a}{\tau}{s}{\tau'} & \wfupd{\Gamma}{a}{\tau'}{s'}{ \tau''}}
    \quad
    \infer{\wfupd{\Gamma}{a}{\tau_1}{\letin{x=e}{s}}{\tau_2}}
    {\wf{\Gamma}{e}{\tau} & \wfupd{\Gamma,x{:}\tau}{a}{\tau_1}{s}{\tau_2}}
    \smallskip\\
    \infer{\wfupd{\Gamma}{a}{\tau}{\ifthenelse{e}{s}{s'}}{\tau_1 | \tau_2}}
    {\wf{\Gamma}{e}{\kbool} & \wfupd{\Gamma}{a}{\tau}{s}{\tau_1} & \wfupd{\Gamma}{a}{\tau}{s'}{\tau_2}}
    \quad
    \infer{\wfupd{\Gamma}{a}{\tau}{\snapshot{x}{s}}{\tau'}}{\wfupd{\Gamma,x{:}\tau}{a}{\tau}{s}{\tau'}}
    \smallskip\\
    \infer{\wfupd{\Gamma}{*}{\emptyseq}{\kinsert~e}{\tau}}
    {\wf{\Gamma}{e}{\tau}}
    \quad
    \infer{\wfupd{\Gamma}{a}{\tau}{\kdelete}{\emptyseq}}
    {}
    \quad
    \infer{\wfupd{\Gamma}{1}{n'[\tau]}{\krename~n}{n[\tau]}}
    {}
    \smallskip\\
    \infer{\wfupd{\Gamma}{1}{\alpha}{\phi?s}{\tau}}
    {\alpha \subty \phi & \wfupd{\Gamma}{1}{\alpha}{s}{\tau}}
    \quad
    \infer{\wfupd{\Gamma}{1}{\alpha}{\phi?s}{\alpha}}
    {\alpha \not\subty \phi}
    \quad
    \infer{\wfupd{\Gamma}{1}{n[\tau]}{\kchildren[s]}{n[\tau']}}
    {\wfupd{\Gamma}{*}{\tau}{s}{\tau'}}
    \smallskip\\
    \infer{\wfupd{\Gamma}{a}{\tau}{\kleft[s]}{\tau',\tau}}
    {\wfupd{\Gamma}{*}{\emptyseq}{s}{\tau'}}
    \quad
    \infer{\wfupd{\Gamma}{a}{\tau}{\kright[s]}{\tau,\tau'}}
    {\wfupd{\Gamma}{*}{\emptyseq}{s}{\tau'}}
    \quad
    \infer{\wfupd{\Gamma}{*}{\tau}{\kiter[s]}{\tau'}}
    {\wfiter{\Gamma}{\tau}{s}{\tau'}}
    \smallskip\\
    \infer{\wfupd{\Gamma}{a}{\tau_1}{s}{\tau_2}}
    { 
      \wfupd{\Gamma}{a}{\tau_1}{s}{\tau_2'} & 
      \tau_2' \subty \tau_2}
    \quad
    \infer{\wfupd{\Gamma}{a}{\sigma_1}{P(\vec{e})}{\sigma_2}}
    {
      \procdecl{P(\vec{\tau})}{\sigma}{\sigma_2}\in \Delta & 
      \sigma_1 \subty \sigma &
      \wf{\Gamma}{\vec{e}}{\vec{\tau}}
    }
  \end{array}\]
  \fbox{$\wfiter{\Gamma}{\tau}{s}{\tau'}$}
  \[\begin{array}{c}
    \infer{\wfiter{\Gamma}{\emptyseq}{s}{\emptyseq}}{}
    \quad
    \infer{\wfiter{\Gamma}{\alpha}{s}{\tau}}
    {\wfupd{\Gamma}{1}{\alpha}{s}{\tau}}
\quad
    \infer{\wfiter{\Gamma}{X}{s}{\tau}}
    {\wfiter{\Gamma}{E(X)}{s}{\tau}}
    \quad
    \infer{\wfiter{\Gamma}{\tau_1^*}{s}{\tau_2^*}}
    {\wfiter{\Gamma}{\tau_1}{s}{\tau_2}}
    \smallskip\\
    \infer{\wfiter{\Gamma}{\tau_1,\tau_2}{s}{\tau_1',\tau_2'}}
    {\wfiter{\Gamma}{\tau_1}{s}{\tau_1'}
      &
      \wfiter{\Gamma}{\tau_2}{s}{\tau_2'}}
\quad
    \infer{\wfiter{\Gamma}{\tau_1|\tau_2}{s}{\tau_1'|\tau_2'}}
    {\wfiter{\Gamma}{\tau_1}{s}{\tau_1'}
      &
      \wfiter{\Gamma}{\tau_2}{s}{\tau_2'}}
  \end{array}\]
  \fbox{$\wfprog{\Gamma}{p}$}
\[
\begin{array}{c}
\infer{\wfprog{\Gamma}{\update{s}{\tau_1}{\tau_2}}}{\wfupd{\Gamma}{*}{\tau_1}{s}{\tau_2}}
\quad
    \infer{\wfprog{\Gamma}{\declareprocedure{\procdecl{P(\vec{x}~{:}~\vec{\tau})}{\tau_1}{\tau_2}}{s}{p}}}{
\begin{array}{c}
\text{$P$ not declared in $p$}\\
      \procdecl{P(\vec{\tau})}{\sigma_1}{\sigma_2}\in \Delta  \quad
        \wfupd{\Gamma,\vec{x}{:}\vec{\tau} }{*}{\sigma_1}{s}{\sigma_2} \quad 
      \wfprog{\Gamma}{p}
\end{array}
}
  \end{array}\]
    \caption{Update and additional program well-formedness rules}\labelFig{update-program-wf}
\end{figure}

In typechecking updates, we extend the global declaration context
$\Delta$ with procedure declarations:
\[
\Delta ::= \cdots \mid \Delta,\procdecl{P(\vec{\tau})}{\tau_1}{\tau_2}
\]
There are two typing judgments for updates: singular well-formedness
$\wfupd{\Gamma}{1}{\alpha}{s}{\tau'}$ (that is, in type environment
$\Gamma$, update $s$ maps tree type $\alpha$ to type $\tau'$), and
plural well-formedness $\wfupd{\Gamma}{*}{\tau}{s}{\tau'}$ (that is,
in type environment $\Gamma$, update $s$ maps type $\tau$ to type
$\tau'$).  Several of the rules are parameterized by a multiplicity $a
\in \{1,*\}$.  In addition, there is an auxiliary judgment
$\wfiter{\Gamma}{\tau}{s}{\tau'}$ for typechecking iterations.  The
rules for update well-formedness are shown in \refFig{update-program-wf}.  We
also need an auxiliary subtyping relation involving atomic types and
tests: we say that $\alpha \subty \phi$ if $\SB{\alpha} \subseteq
\SB{\phi}$.  This is characterized by the  rules:
\[\infer{\kbool \subty \kbool}{}
\quad
\infer{\kstring \subty \kstring}{}
\quad 
\infer{n[\tau] \subty  n}{}\quad \infer{n[\tau] \subty *}{}\]

\begin{remark}
In most other XML update proposals (including
XQuery!~\cite{DBLP:conf/edbtw/GhelliRS06} and the draft XQuery Update
Facility~\cite{xquery-update-w3c-072006}), side-effecting update
operations are treated as \emph{expressions} that return $\emptyseq$.
Thus, we could perhaps typecheck such updates as expressions of type
$\emptyseq$.  This would work fine as long as the types of values
reachable from the free variables in $\Gamma$ can never change;
however, the updates available in these languages can and do change
the values  of variables.  Thus, to make this approach
sound $\Gamma$ would to be updated to take these changes into account,
perhaps using a judgment $\Gamma \vdash s : \emptyseq \mid \Gamma'$,
where $\Gamma'$ is the updated type environment reflecting the types
of the variables after update $s$.  This approach quickly becomes
difficult to manage, especially if it is possible for different
variables to ``alias'', or refer to overlapping parts of the data
accessible from $\Gamma$, and adding side-effecting functions further
complicates matters.

This is \emph{not} the approach to update typechecking that is taken
in \Flux.  Updates are syntactically distinct from queries, and a
\Flux update typechecking judgment such as
$\wfupd{\Gamma}{a}{\tau}{s}{\tau'}$ assigns an update much richer type
information that describes the type of part of the database before and
after running $s$.  The values of variables bound in $\Gamma$ are
immutable in the variable's scope, so their types do not need to be
updated.  Similarly, procedures must be annotated with expected input
and output types.  We do not believe that these annotations are
burdensome in a database setting since a typical update procedure
would be expected to preserve the (usually fixed) type of the
database.
\end{remark}

\subsection{Examples}\labelSec{update-examples}

The interesting rules are those involving $\kiter$, tests, and
$\kchildren$, $\kleft/\kright$, and $\kinsert/\krename/\kdelete$.  The
following example should help illustrate how the rules work for these
constructs.  Consider the high-level update:
\begin{verbatim}
insert after a/b value c[]
\end{verbatim}
which can be compiled to the following core \Flux statement:
\[\kiter~[a?\kchildren~[\kiter~[b?~\kright~\kinsert~c[]]]]\]
Intuitively, this update inserts a $c$ after every $b$ under a
top-level $a$.  Now consider the input type $a[b[]^*,c[]],d[]$.
Clearly, the output type \emph{should} be $a[(b[],c[])^*,c[]],d[]$.
To see how \Flux can assign this type to the update, consider the
derivation shown in \refFig{update-example}.

\begin{figure}[tb]
\[
\infer{\wfupd{}{*}{a[b[]^*,c[]],d[]}{\kiter~[a?\kchildren[s]]}{a[(b[],c[])^*,c[]],d[]}}
{
  \infer{\wfiter{}{a[b[]^*,c[]],d[]}{a?\kchildren[s]}{a[(b[],c[])^*,c[]],d[]}}
  {
    \infer{\wfiter{}{a[b[]^*,c[]]}{a?\kchildren[s]}{a[(b[],c[])^*,c[]]}}
    {
      \infer{\wfupd{}{1}{a[b[]^*,c[]]}{\kchildren[s]}{a[(b[],c[])^*,c[]]}}
      {
        \infer{\wfupd{}{*}{b[]^*,c[]}{\kiter~[b?s']}{(b[],c[])^*,c[]}}
        {
          \infer{\wfiter{}{b[]^*,c[]}{b?s'}{(b[],c[])^*,c[]}}
          {
            \infer{\wfiter{}{b[]^*}{b?s'}{(b[],c[])^*}}
            {
              \infer{\wfiter{}{b[]}{b?s'}{b[],c[]}}
              {
                \infer{\wfupd{}{1}{b[]}{b?s'}{b[],c[]}}
                {
                  \infer{\wfupd{}{1}{b[]}{\kright~\kinsert~c[]}{b[],c[]}}
                  {
                    \infer{\wfupd{}{*}{\emptyseq}{\kinsert~c[]}{b[],c[]}}
                    {
                      \hyp{\wf{}{c[]}{c[]}}
                    }
                  }
                }
              }
            }
          }
          &
          \infer{\wfiter{}{c[]}{b?s'}{c[]}}
          {
            \infer{\wfupd{}{1}{c[]}{b?s'}{c[]}}{}
          }
        }
      }
    }
    & 
    \infer{\wfiter{}{d[]}{a?\kchildren[s]}{d[]}}
    {}
  }
}
\]
\caption{Example update derivation, where $s' = \kright~\kinsert~c[]$
  and $s = \kiter~[b?s']$}\labelFig{update-example}
\end{figure}
\begin{figure}[tb]
\[
\infer{\wfupd{x{:}\kstring}{*}{\node[\Tree^*]}{\node?\kchildren[\kiter[\leafupd(x)]]}{\node[\Tree^*]}}
{\infer{\wfupd{x{:}\kstring}{*}{\node[\Tree^*]}{\kchildren[\kiter[\leafupd(x)]]}{\node[\Tree^*]}}
  {\infer{\wfupd{x{:}\kstring}{*}{\Tree^*}{\kiter[\leafupd(x)]}{\Tree^*}}
    {\infer{\wfiter{x{:}\kstring}{\Tree^*}{\leafupd(x)}{\Tree^*}}
      {\infer{\wfiter{x{:}\kstring}{\Tree}{\leafupd(x)}{\Tree}}
        {\infer{\wfiter{x{:}\kstring}{\tree[\leaf[\kstring]|\node[\Tree^*]]}{\leafupd(x)}{\Tree}}
          {\infer{\wfupd{x{:}\kstring}{1}{\tree[\leaf[\kstring]|\node[\Tree^*]]}{\leafupd(x)}{\Tree}}
            {
              \procdecl{\leafupd(\kstring)}{\Tree}{\Tree} \in \Delta
              & 
              \tree[...] \subty \Tree
              & 
              \wf{x{:}\kstring}{x}{\kstring}
            }
          }
        }
      }
    }
  }
}
\]
\caption{Partial derivation for declaration of $\leafupd$}
\labelFig{update-example2}
\end{figure}

As a second example, consider the procedure declaration
\[
\begin{array}{l}
  \declareprocedure{\procdecl{\leafupd(x{:}\kstring)}{\Tree}{\Tree}}{\\
\qquad\kiter[\kchildren[\kiter[\leaf?\kchildren[\kdelete;\kinsert~x];\\
\qquad\qquad\qquad\qquad~~~\node?\kchildren[\kiter[\leafupd(x)]]]]]\\
}
\end{array}
\]
This procedure updates all leaves of a tree to $x$.  As with the
recursive query discussed in \refSec{query-examples}, this procedure requires
subtyping to typecheck the recursive call.  We also need subtyping to
check that the return type of the expression matches the declaration.
A partial typing derivation for part of the body of the procedure
involving a recursive call is shown in \refFig{update-example2}.

\subsection{Decidability}\labelSec{update-decidability}

To decide typechecking, we must again carefully control the use of
subsumption. The appropriate algorithmic typechecking judgment is
defined as follows:
\begin{definition}[Algorithmic derivations for updates]
  The algorithmic typechecking judgments
  $\wfupdalg{\Gamma}{a}{\tau}{s}{\tau'}$ and
  $\wfiteralg{\Gamma}{\tau}{s}{\tau'}$ are obtained by taking the
  rules in \refFig{update-program-wf}, removing both subsumption
  rules, and replacing the procedure call rule with
\[
    \infer{\wfupdalg{\Gamma}{a}{\tau}{P(\vec{e})}{\sigma'}}
    {
      \procdecl{P(\vec{\sigma})}{\sigma}{\sigma'}\in \Delta & 
      \tau \subty \sigma & 
      \wfalg{\Gamma}{\vec{e}}{\vec{\tau}}
      & \vec{\tau} \subty \vec{\sigma}
    }
\]
Moreover, all subderivations of expression judgments in an algorithmic
derivation of an update judgment must be algorithmic.
\end{definition}

The proof of completeness of algorithmic update typechecking has the
same structure as that for queries.  We state the main results; proof
details are in the appendix.

\begin{lemma}[Decidabilty for updates]\labelLem{decidability-upd}
  Let $a,s$ be given.  Then there exist computable
  functions $j_{a,s}$ and $k_s$ such that:
  \begin{enumerate}
  \item $j_{a,s}(\Gamma,\tau)$ is the unique $\tau_2$ such that
    $\wfupdalg{\Gamma}{a}{\tau_1}{s}{\tau_2}$, if it exists.
  \item $k_{s}(\Gamma,\tau_1)$ is the unique $\tau_2$ such that
    $\wfiteralg{\Gamma}{\tau_1}{s}{\tau_2}$, if it exists.
  \end{enumerate}
\end{lemma}
\begin{theorem}[Algorithmic soundness for updates]\labelThm{soundness-upd}
  (1) If $\wfupdalg{\Gamma}{*}{\tau}{s}{\tau'}$ is derivable then $\wfupdalg{\Gamma}{*}{\tau}{s}{\tau'}$ is derivable.
  (2) If $\wfiteralg{\Gamma}{\tau}{e}{\tau'}$ is derivable then
  $\wfiter{\Gamma}{\tau}{e}{\tau'}$ is derivable.
\end{theorem}

\begin{lemma}[Downward monotonicity for updates]
  \labelLem{downward-monotonicity-upd}
(1) If $\wfupdalg{\Gamma}{a}{\tau_1}{s}{\tau_2}$ and $\Gamma'
    \subty \Gamma$ and $\tau_1' \subty \tau_1$ then
    $\wfupdalg{\Gamma'}{a}{\tau_1'}{s}{\tau_2'}$ for some $\tau_2'
    \subty \tau_2$.
(2) If $\wfiteralg{\Gamma}{\tau_1}{s}{\tau_2}$ and $\Gamma' \subty
    \Gamma$ and $\tau_1' \subty \tau_1$ then
    $\wfiteralg{\Gamma'}{\tau_1'}{s}{\tau_2'}$ for some $\tau_2'
    \subty \tau_2$.
\end{lemma}

\begin{theorem}[Algorithmic completeness for updates]\labelThm{alg-completeness-updates}
(1) If $\wfupd{\Gamma}{a}{\tau_1}{s}{\tau_2}$ then there exists
    $\tau_2' \subty \tau_2$ such that
    $\wfupdalg{\Gamma}{a}{\tau_1}{s}{\tau_2'}$.
(2) If $\wfiter{\Gamma}{\tau_1}{s}{\tau_2}$ then there exists
    $\tau_2' \subty \tau_2$ such that
    $\wfiteralg{\Gamma}{\tau_1}{s}{\tau_2'}$.
\end{theorem}

\section{Related and future work}\labelSec{related-and-future-work}

This work is directly motivated by our interest in using regular
expression types for XML updates, using richer typing rules for
iteration as found in \muXQ~\cite{colazzo06jfp}.  Fernandez, Sim\'eon
and Wadler~\cite{fernandez01icdt} earlier considered an XML query
language with more precise typechecking for iteration, but this
proposal required many more type annotations than XQuery, \muXQ or
\Flux do; we only require annotations on function or procedure
declarations.

For brevity, the core languages in this paper omitted many features of
full XQuery, such as the descendant, attribute, parent and sibling axes.
The attribute axis is straightforward since attributes always have
text content.  In \muXQ, the descendant axis was supported by
assigning $\bar{x}/\text{\texttt{descendant-or-self}}$ the type formed
by taking the union of all tree types that are reachable from the type
of $\bar{x}$.  XQuery handles other axes by discarding type
information.  Our algorithmic completeness proof still appears to work
if these axes are added.

We are also interested in extending the path
correctness analysis introduced by Colazzo et al.  to \Flux.  In the
update setting, a natural form of path correctness might be that there
are no statically ``dead'' updates.

\Flux represents a fundamental departure from the other XML update
language proposals of which we are aware (such as
XQuery!~\cite{ghelli07icdt} and the draft W3C XQuery Update
Facility~\cite{xquery-update-w3c-072006}).  To the best of our
knowledge, static typechecking and subtyping have yet to be considered
for such languages and seem likely to encounter difficulties for
reasons we outlined in \refSec{update-types} and discussed in more
depth in~\cite{cheney07planx}.  In addition, \Flux satisfies many
algebraic laws that can be used to rewrite updates without first
needing to perform static analysis, whereas a sophisticated analysis
needs to be performed in XQuery! even to determine whether two query
expressions can be reordered.  We believe that this will enable
aggressive update optimizations.

On the other hand, XQuery! and related proposals are clearly more
expressive than \Flux, and have been incorporated into XML database
systems such as Galax~\cite{DBLP:conf/vldb/FernandezSCMS03}.  Although
we currently have a prototype that implements the typechecking
algorithm described here as well as the operational semantics
described in \cite{cheney07planx}, further work is needed to develop a
robust implementation inside an XML database system that could be used
to compare the scalability and optimizability of \Flux with other
proposals.

\section{Conclusions}\labelSec{concl}

Static typechecking is important in a database setting because type
(or ``schema'') information is useful for optimizing queries and
avoiding expensive run-time checks or re-validation.  The XQuery
standard, like other XML programming languages, employs regular
expression types and subtyping.  However, its approach to typechecking
iteration constructs is imprecise, due to the use of ``factoring''
which discards information about the order of elements in the result
of an iteration operation such as a $\kfor$-loop.  While this
imprecision may not be harmful for typical queries, it is disastrous
for typechecking updates that are supposed to preserve the type of the
database.

In this paper we have considered more precise typing disciplines for
XQuery-style iterative queries and updates in the core languages \muXQ
and \Flux respectively.  In order to ensure that these type systems
are well-behaved and that typechecking is decidable, it is important
to prove the completeness of an algorithmic presentation of
typechecking in which the use of subtyping rules is limited so that
typechecking remains syntax-directed.  We have shown how to do so for
the core \muXQ and \Flux languages, and believe the proof technique
will extend to handle other features not included in the paper.  These
results provide a solid foundation for subtyping in XML query and
update languages with precise iteration typechecking rules and for
combining them with other XML programming paradigms based on regular
expression types.

\bibliographystyle{plain}
\bibliography{fp,db,xml,xupdate,paper}

\newpage
\appendix

\input{appendix}
\end{document}

%% file: nominal.tex
\usepackage{xspace}
\usepackage{graphicx}

\usepackage{amssymb}
\usepackage{amsfonts}
\usepackage{proof}
\usepackage{epsfig}
\usepackage{verbatim}

\newcommand{\SB}[1]{[\![#1]\!]}

\newcommand{\eq}{\approx}

\newcommand{\hyp}[2][]{\infer[#1]{#2}{}}

\newcommand{\labelSec}[1]{\label{sec:#1}}
\newcommand{\refSec}[1]{Section~\ref{sec:#1}}
\newcommand{\labelFig}[1]{\label{fig:#1}}
\newcommand{\refFig}[1]{Figure~\ref{fig:#1}}

\newcommand{\labelThm}[1]{\label{thm:#1}}
\newcommand{\refThm}[1]{Theorem~\ref{thm:#1}}
\newcommand{\labelLem}[1]{\label{lem:#1}}
\newcommand{\refLem}[1]{Lemma~\ref{lem:#1}}
\newcommand{\labelProp}[1]{\label{prop:#1}}

\newcommand{\To}{\Rightarrow}
\newcommand{\eval}{\Downarrow}

\newcommand{\PP}{\mathcal{P}}

\newcommand{\nd}{\vdash}


%% file: macros.tex
\newcommand{\DD}{\mathcal{D}}

\newcommand{\evalu}[4]{#1;#2 \vdash #3 \Rightarrow^\mathtt{U} #4}
\renewcommand{\eval}[3]{#1 \vdash #2 \Rightarrow #3}

\newcommand{\pto}{\rightharpoonup}
\newcommand{\subty}{\mathrel{{<}{:}}}
\newcommand{\supty}{\mathrel{{>}{:}}}

\newcommand{\muXQ}{$\mu$\textsc{xq}\xspace}
\newcommand{\Flux}{\textsc{Flux}\xspace}

\newcommand{\vdasharrow}{\mathrel{\vdash\!\!\!\!\text{\tiny$^\blacktriangleright$}}}
\newcommand{\tylab}[3]{#1 :: #2 \To #3}
\newcommand{\wfprog}[2]{#1 \vdash #2 ~ \mathtt{prog}}
\newcommand{\wfin}[5]{#1 \vdash \bar{#2}~\kin~#3 \to #4 : #5}
\newcommand{\wfinalg}[5]{#1 \vdasharrow \bar{#2}~\kin~#3 \to #4 : #5}
\newcommand{\wfiter}[4]{#1 \vdash_{\!\!\kiter} \{#2\}~ #3~ \{#4\}}
\newcommand{\wfiteralg}[4]{#1 \vdasharrow_{\!\!\kiter} \{#2\}~ #3~ \{#4\}}
\newcommand{\wfupd}[5]{#1 \vdash^{\!\!{ \tiny #2}} \{#3\}~#4~\{#5\}}
\newcommand{\wfupdalg}[5]{#1 \vdasharrow^{\!\!{\tiny  #2}} \{#3\}~#4~\{#5\}}

\newcommand{\wf}[3]{#1 \vdash #2:#3}
\newcommand{\wfalg}[3]{#1 \vdasharrow #2:#3}

\newcommand{\Atoms}{\mathit{Atom}}
\newcommand{\Stmt}{\mathit{Stmt}}

\newcommand{\Lab}{\mathit{Lab}}
\newcommand{\Test}{\mathit{Test}}
\newcommand{\Expr}{\mathit{Expr}}
\newcommand{\FSym}{\mathit{FSym}}
\newcommand{\PSym}{\mathit{PSym}}
\newcommand{\Prog}{\mathit{Prog}}
\newcommand{\Dir}{\mathit{Dir}}
\newcommand{\Types}{\mathit{Type}}
\newcommand{\Bool}{\mathit{Bool}}
\newcommand{\Val}{\mathit{Val}}
\newcommand{\Var}{\mathit{Var}}
\newcommand{\TVar}{\mathit{TVar}}

\newcommand{\emptyseq}{\texttt{()}}
\newcommand{\kw}[1]{\mathtt{#1}}

\newcommand{\kchildren}{\kw{children}}
\newcommand{\kchild}{\kw{child}}

\newcommand{\ksnapshot}{\kw{snapshot}}
\newcommand{\snapshot}[2]{\ksnapshot~#1~\kin~#2}
\newcommand{\kbool}{\kw{bool}}
\newcommand{\kstring}{\kw{string}}
\newcommand{\kreturn}{\kw{return}}
\newcommand{\ktrue}{\kw{true}}
\newcommand{\kfalse}{\kw{false}}
\newcommand{\kdeclare}{\kw{declare}}
\newcommand{\kfunction}{\kw{function}}
\newcommand{\kprocedure}{\kw{procedure}}
\newcommand{\kint}{\kw{int}}

\newcommand{\kleft}{\kw{left}}
\newcommand{\kright}{\kw{right}}

\newcommand{\kinsert}{\kw{insert}}

\newcommand{\kdelete}{\kw{delete}}

\newcommand{\krename}{\kw{rename}}
\newcommand{\kif}{\kw{if}}
\newcommand{\kthen}{\kw{then}}
\newcommand{\kelse}{\kw{else}}
\newcommand{\klet}{\kw{let}}
\newcommand{\kin}{\kw{in}}

\newcommand{\kskip}{\kw{skip}}

\newcommand{\kfor}{\kw{for}}

\newcommand{\kiter}{\kw{iter}}

\newcommand{\ifthenelse}[3]{\kif~#1~\kthen~#2~\kelse~#3}

\newcommand{\declarefunction}[3]{\kdeclare~\kfunction~ #1~\{#2\}\text{\texttt{;}}~#3}
\newcommand{\declareprocedure}[3]{\kdeclare~\kprocedure~ #1~\{#2\}\text{\texttt{;}}~#3}
\newcommand{\letin}[2]{\klet~ #1~\kin~#2}

\newcommand{\forreturn}[2]{\kfor~#1~\kreturn~#2}

\renewcommand{\vec}[1]{\overline{#1}}

\newcommand{\ichildren}{\mathit{children}}
\newcommand{\procdecl}[3]{#1:#2\To#3}
\newcommand{\funcdecl}[2]{#1:#2}
\newcommand{\typedecl}[2]{\mathtt{type}~#1=#2}
\newcommand{\query}[2]{\mathtt{query}~#1:#2}
\newcommand{\update}[3]{\mathtt{update}~#1:#2\To #3}
\newcommand{\leaf}{\mathit{leaf}}
\newcommand{\leaves}{\mathit{leaves}}
\newcommand{\leafupd}{\mathit{leafupd}}
\newcommand{\node}{\mathit{node}}
\newcommand{\Tree}{\mathit{Tree}}
\newcommand{\tree}{\mathit{tree}}
\newcommand{\dom}{\mathrm{dom}}


%% file: abstract.tex
XML database query languages such as XQuery employ regular	
expression types with structural subtyping.  Subtyping systems
typically have two presentations, which should be equivalent: a
declarative version in which the subsumption rule may be used
anywhere, and an algorithmic version in which the use of subsumption
is limited in order to make typechecking syntax-directed and
decidable.  However, the XQuery standard type system circumvents
this issue by using imprecise typing rules for iteration
constructs and defining only algorithmic typechecking, and another
extant proposal provides more precise types for iteration constructs
but ignores subtyping.  In this paper, we consider a core
XQuery-like language with a subsumption rule and prove the
completeness of algorithmic typechecking; this is straightforward
for XQuery proper but requires some care in the presence of more
precise iteration typing disciplines.  We extend this result to an
XML update language we have introduced in earlier work.

%% file: appendix.tex
\newpage
\appendix

\section{Proofs from Sections~\ref{sec:query-decidability} and
  \ref{sec:update-decidability}}

\subsection{Regular languages and homomorphisms}

We assume familiarity with the theory of regular expressions and
regular languages; in this case, we consider types $\tau \in \Types$
to be regular languages over \emph{atomic types} $\alpha \in \Atoms$.
The language $L(\tau)$ denoted by a type is therefore a set of
sequences $\omega \in \Atoms^*$ of atomic types, where $L:\Types \to
\Atoms^*$ is defined as follows:
\begin{eqnarray*}
L(\emptyseq) &=& \{\emptyseq\}\\
L(\alpha) &=& \{\alpha' \mid \alpha' \subty \alpha\}\\
L(\tau,\tau') &=& L(\tau)\bullet L(\tau')= \{\omega,\omega' \mid \omega \in L(\tau), \omega' \in L(\tau') \}\\
L(\tau | \tau') &=& L(\tau) \cup L(\tau')\\
L(\tau^*) &=& L(\tau)^* = \bigcup_{i=0}^\infty L(\tau)^n\\
L(X) &=& L(E(X))
\end{eqnarray*}
Note that this definition differs slightly from the usual definition
of the language of a regular expression, in that we include all
subtypes of atomic types $\alpha$ in $L(\alpha)$.

It is straightforward to show the following useful properties of $L$:
\begin{lemma}\labelLem{l-is-atomic-subtypes}
  $L(\tau) = \{\omega \mid \omega \subty \tau\}$
\end{lemma}
\begin{proof}
  For both directions, proof is by induction on the structure of
  $\tau$.  For the forward direction, we have:
\begin{itemize}
\item Case \emptyseq: immediate
\item Case $\alpha$: Suppose $\omega \in L(\alpha)$.  Clearly $\omega = \alpha' \subty \alpha$ for some atomic $\alpha'$.
\item Case $\tau_1,\tau_2$: Suppose $\omega \in L(\tau_1,\tau_2)$.  By
  definition, $\omega = \omega_1 , \omega_2$ where $\omega_i \in
  L(\tau_i)$ for $i \in \{1,2\}$.  Then by induction $\omega_i \subty
  \tau_i$ for $i \in \{1,2\}$.  Thus $\omega_1 ,
  \omega_2\subty\tau_1,\tau_2$.
\item Case $\tau_1|\tau_2$: Suppose $\omega \in L(\tau_1|\tau_2)$.  By
  definition, $\omega = \omega_i $ where $\omega \in L(\tau_i)$ for
  some $i \in \{1,2\}$.  Then by induction $\omega \subty \tau_i$ for
  some $i \in \{1,2\}$.  Thus $\omega\subty\tau_1|\tau_2$.
\item Case $\tau^*$: Suppose $\omega \in L(\tau^*)$. By definition,
  $\omega = \omega_1,\ldots,\omega_n $ where $n \geq 0$ and $\omega_i
  \in L(\tau)$ for all $i \in \{1,\ldots,n\}$.  Then by induction
  $\omega_i \subty \tau$ for all $i \in \{1,\ldots,n\}$.  Thus $\omega
  = \omega_1,\ldots,\omega_n\subty\tau,\ldots,\tau \subty \tau^*$.
\item Case $X$: Immediate by induction.
\end{itemize}
For the reverse direction, we have:
\begin{itemize}
\item Case \emptyseq: immediate, since we must have $\omega = \emptyseq \in L(\emptyseq)$
\item Case $\alpha$: Suppose $\omega \subty \alpha$.  Clearly $\omega
  = \alpha' \subty \alpha$ for some atomic $\alpha'$, so $\omega \in
  L(\alpha)$.
\item Case $\tau_1,\tau_2$: Suppose  $\omega \subty \tau_1,\tau_2$.  Then since $\omega$ is atomic we must have $\omega = \omega_1,\omega_2$ where $\omega_i \subty \tau_i$ for $i \in \{1,2\}$.  Thus $\omega = \omega_1,\omega_2 \in L(\tau_1)\bullet L(\tau_2) = L(\tau_1,\tau_2)$.
\item Case $\tau_1|\tau_2$: Since $\omega$ is atomic, $\omega \subty \tau_1 | \tau_2$ implies that $\omega \subty \tau_1$ or $\omega \subty \tau_2$.  Thus $\omega \in L(\tau_1) \cup L(\tau_2) = L(\tau_1|\tau_2)$.
\item Case $\tau^*$: Since $\omega$ is atomic, we must have $\omega =
  \omega_1,\ldots, \omega_n$ where $\omega_i \subty \tau$; hence $\omega =
  \omega_1,\ldots, \omega_n \in L(\tau)^* = L(\tau^*)$.
\item Case $X$: Immediate by induction.
\end{itemize}
\end{proof}

\begin{lemma}\labelLem{values-have-atomic-witnesses}
  If $v \in \SB{\tau}$, then there exists a $\omega \in L(\tau)$ such
  that $v \in \SB{\omega}$.
\end{lemma}
\begin{proof}
  Induction on the structure of $v,\tau$.
\begin{itemize}
\item Case $\emptyseq,\emptyseq$: Immediate; $\omega = \emptyseq$ works.
\item Case $\bar{v},\alpha$: Immediate; $\omega$ = $\alpha$ works.
\item Case $v,(\tau_1,\tau_2)$: We must have $v = v_1,v_2$ where $v_i \in \SB{\tau_i}$, for $i \in \{1,2\}$.  Then by induction we have $\omega_i \in L(\tau_i)$ with $v_i \in \SB{\omega_i}$; this implies $v \in \SB{\omega_1,\omega_2} \subseteq \SB{\tau_1,\tau_2}$.
\item Case $v,\tau_1|\tau_2$: Without loss of generality, suppose $v \in \SB{\tau_i}$.  Then by induction we have $\omega \in L(\tau_i) \subseteq L(\tau_1|\tau_2)$ such that $v \in \SB{\omega} \subseteq \SB{\tau_1 | \tau_2}$.
\item Case $v,\tau^*$: If $v = \emptyseq$, then $\omega = \emptyseq$ works.
  Otherwise we must have $v = v_1,\ldots,v_n$ where $v_i \in
  \SB{\tau}$.  Then by induction we have $\omega_i \in L(\tau)$ with
  $v_i \in \SB{\omega}$; this implies that $\omega_1,\ldots,\omega_n \in L(\tau^*)$ and $v \in \SB{\omega_1,\ldots,\omega_n} \subseteq \SB{\tau^*}$.
\item Case $X$: Immediate by induction.
\end{itemize}
\end{proof}

\begin{lemma}\labelLem{l-monotonic}
  For any $\tau,\tau'$, $\tau \subty \tau'$ if and only if $ L(\tau)
  \subseteq L(\tau')$
\end{lemma}
\begin{proof} 
  In the forward direction, if $\tau \subty \tau'$, then let $\omega
  \in L(\tau)$ be given.  Then $\omega \subty \tau \subty \tau'$.
  Thus, $\omega \in L(\tau')$.

  In the reverse direction, suppose that $L(\tau) \subseteq L(\tau')$.
  Suppose $v \in \SB{\tau}$.  Via
  \refLem{values-have-atomic-witnesses}, choose $\omega$ such that $v
  \in \SB{\omega}$ and $\omega \in L(\tau)$.  Since $L(\tau) \subseteq
  L(\tau')$, we have that $\omega \subty \tau'$, so $v \in \SB{\omega}
  \subseteq \SB{\tau'}$.  We conclude that $\SB{\tau} \subseteq
  \SB{\tau'}$ so by definition $\tau \subty \tau'$.
\end{proof}

We now recall properties of homomorphisms of regular type expressions.
A (partial) homomorphism $h : \Types \to \Types$ (or $h : \Types \pto
\Types$) is a (partial) function satisfying
\begin{eqnarray*}
h(\emptyseq) &=& \emptyseq\\
h(\tau,\tau') &=& h(\tau),h(\tau')\\
h(\tau|\tau') &=& h(\tau)|h(\tau')\\
h(\tau^*) &=& h(\tau)^*\\
h(X) &=& h(E(X))
\end{eqnarray*}
In particular, we consider (partial) homomorphisms that are generated
entirely by their behavior on atoms, that is, given a (partial)
function $k : \Atoms \to \Types$, we construct the unique (partial)
homomorphism $\hat{k}$ agreeing with $k$ by taking $\hat{k}(\alpha) =
k(\alpha)$ (when defined) and using the above equations in all other
cases.  

We say that a (partial) function $F : X \pto Y$ on ordered sets $X,Y$
is downward closed if whenever $x' \leq_X x$, and $F(x)$
exists, then $F(x')$ also exists; a downward closed function is
\emph{downward monotonic} if in addition $F(x') \leq_Y F(x)$.

In the following, we use the notation $F[-] : \PP(X) \pto \PP(Y)$ for
the partial function on sets obtained by lifting $F : X \pto Y$;
$F[S]$ is defined and equals $\{F(s) \mid s \in S\}$ provided $F$ is
defined on each element of $S$. It is easy to show that this operation
is downward monotonic with respect to set inclusion and preserves
totality (if $F$ is total then $F[-]$ is total also).

We need a second auxiliary function, namely the set of atoms appearing in 
a type.  This is given by $A : \Types \to \PP(\Atoms)$, defined as follows:
\begin{eqnarray*}
A(\emptyseq) &=& \{\}\\
A(\alpha) &=& \{\alpha' \mid \alpha' \subty \alpha\}\\
A(\tau,\tau') &=& A(\tau)\cup A(\tau')\\
A(\tau | \tau') &=& A(\tau) \cup A(\tau')\\
A(\tau^*) &=& A(\tau)\\
A(X) &=& A(E(X))
\end{eqnarray*}
The following fact about $A$ will be needed:
\begin{lemma}\labelLem{a-monotonic}
  If $\tau \subty \tau'$ then $A(\tau) \subseteq A(\tau')$.
\end{lemma}
\begin{proof}
  Note that $A(\tau) = \bigcup B[L(\tau)]$ where $B : \Atoms^* \to
  \PP(\Atoms)$ is defined by 
\begin{eqnarray*}
B(\emptyseq) &=& \{\}\\
B(\alpha\omega)&  =&\{\alpha' \mid \alpha' \subty \alpha\} \cup B(\omega)
\end{eqnarray*}
and $\bigcup : \PP(\PP(\Atoms)) \to \PP(\Atoms)$ is the usual
flattening operator on sets.  All three functions $\bigcup,B[-],L$ are
monotonic.
\end{proof}

\begin{lemma}\labelLem{h-generated-by-a}
  Let $h : \Atoms \pto \Types$ be given.  If $h(\alpha)$ is defined
  for each $\alpha \in A(\tau)$ then $\hat{h}(\tau)$ is defined.
\end{lemma}
\begin{proof}
  By structural induction on $\tau$.  The base case $\tau = \alpha$ is
  by definition of $\hat{h}(\alpha) = h(\alpha)$. The remaining cases are
  straightforward because $\hat{h}$ is a homomorphism.
\end{proof}
\begin{lemma}\labelLem{dc-h-defined-on-a}
  If $h: \Atoms \pto \Types$ is downward closed, and $\hat{h}(\tau)$ is
  defined, then $h(\alpha)$ is defined for every $\alpha \in A(\tau)$.
\end{lemma}
\begin{proof}
  By structural induction on $\tau$.  For the base case $\tau =
  \alpha$, we need downward closedness to conclude that $h(\alpha)$
  is defined for each $\alpha' \subty \alpha$.  The remaining cases are
  straightforward because $\hat{h}$ is a homomorphism.
\end{proof}

\begin{lemma}\labelLem{dc-hom}
  If $h:\Atoms\pto \Types$ is downward
  closed, then $\hat{h}$ is downward closed.
\end{lemma}
\begin{proof}
  Let $\tau' \subty \tau$ be given such that $\hat{h}(\tau)$ is defined.  Then
  by \refLem{dc-h-defined-on-a}, $h(\alpha)$ is defined on every
  $\alpha \in A(\tau)$.  But $A(\tau') \subseteq A(\tau)$
  (\refLem{a-monotonic}) so by \refLem{h-generated-by-a}, $\hat{h}(\tau')$
   is defined.
\end{proof}

\begin{lemma}\labelLem{commutation}
  Suppose $h: \Atoms \pto \Types$ is downward monotonic.  Then for
  any $\tau \in \dom(\hat{h})$,
  \begin{equation}
    \bigcup L[\hat{h}[L(\tau)]] = L(\hat{h}(\tau)) 
  \end{equation}
\end{lemma}
\begin{proof}
  By induction on the structure of $\tau$.  
\begin{itemize}
\item $\tau = \emptyseq$.  Then 
  \begin{eqnarray*}
    \bigcup L[\hat{h}[L(\emptyseq)]] &=& \bigcup L[\hat{h}[\{\emptyseq\}]]
  = \bigcup L[\{\hat{h}(\emptyseq)\}]
  = \bigcup L[\{\emptyseq\}]\\
  &=& \bigcup \{L(\emptyseq)\}
   = L(\emptyseq)
  = L(\hat{h}(\emptyseq))
\end{eqnarray*}
\item $\tau = \alpha$.  We need to show that $\bigcup L[h[L(\alpha)]]
  = L(h(\alpha)) $.
  \begin{eqnarray*}
    \bigcup L[h[L(\alpha)]] 
    &=& \bigcup L[h[\{\alpha' \mid \alpha' \subty \alpha\}]] \\
    &=& \bigcup L[\{h(\alpha') \mid \alpha' \subty \alpha\}] \\
    &=& \bigcup \{L(h(\alpha')) \mid \alpha' \subty \alpha\} 
  \end{eqnarray*}
  Now since $h$ is downward monotonic and defined on $\alpha$, for
  each $\alpha' \subty \alpha$ we have that $h(\alpha') \subty
  h(\alpha)$.  Thus, $L(h(\alpha')) \subseteq L(h(\alpha))$, so
  $\bigcup \{L(h(\alpha')) \mid \alpha' \subty \alpha\} =
  L(h(\alpha))$, as desired.
\item $\tau = \tau_1,\tau_2$.  Then 
  \begin{eqnarray*}
    \bigcup L[\hat{h}[L(\tau_1,\tau_2)]]
    &=& \bigcup L[\hat{h}[L(\tau_1)\bullet L(\tau_2)]]
    = \bigcup L[\hat{h}[L(\tau_1)]\bullet \hat{h}[L(\tau_2)]]\\
    &= &\bigcup L[\hat{h}[L(\tau_1)]]\bullet L[\hat{h}[L(\tau_2)]]
    = \left(\bigcup L[\hat{h}[L(\tau_1)]]\right)\bullet \left(\bigcup L[\hat{h}[L(\tau_2)]]\right)\\
    &=& L(\hat{h}(\tau_1))\bullet L(\hat{h}(\tau_2))
    = L(\hat{h}(\tau_1),\hat{h}(\tau_2))
    = L(\hat{h}(\tau_1,\tau_2))
\end{eqnarray*}
\item $\tau = \tau_1|\tau_2$.  Then 
  \begin{eqnarray*}
    \bigcup L[\hat{h}[L(\tau_1|\tau_2)]]
    &=& \bigcup L[\hat{h}[L(\tau_1)\cup L(\tau_2)]]
    = \bigcup L[\hat{h}[L(\tau_1)]\cup\hat{h}[L(\tau_2)]]\\
    &= &\bigcup L[\hat{h}[L(\tau_1)]]\cup L[\hat{h}[L(\tau_2)]]
    = \left(\bigcup L[\hat{h}[L(\tau_1)]]\right)\cup\left(\bigcup L[\hat{h}[L(\tau_2)]]\right)\\
    &=& L(\hat{h}(\tau_1))\cup L(\hat{h}(\tau_2))
    = L(\hat{h}(\tau_1)|\hat{h}(\tau_2))
    = L(\hat{h}(\tau_1|\tau_2))
\end{eqnarray*}
\item  $\tau = \tau^*$.
  \begin{eqnarray*}
    \bigcup L[\hat{h}[L(\tau^*)]]
    &=& \bigcup L[\hat{h}[L(\tau)^*]]
    = \bigcup L[\hat{h}[L(\tau_1)]^*]\\
    &= &\bigcup L[\hat{h}[L(\tau_1)]]^*
    = \left(\bigcup L[\hat{h}[L(\tau)]]\right)^*\\
    &=& L(\hat{h}(\tau_1))^*
    = L(\hat{h}(\tau)^*)
    = L(\hat{h}(\tau^*))
\end{eqnarray*}
  \item $\tau = X$: Immediate by induction.
\end{itemize}
\end{proof}
\begin{theorem}\labelThm{dm-hom}
  If $h:\Atoms\pto \Types$ is downward
  monotonic, then $\hat{h}$ is downward monotonic.
\end{theorem}
\begin{proof}
  Let $\tau' \subty \tau$ be given such that $\hat{h}(\tau)$ is
  defined.  By \refLem{dc-hom}, $\hat{h}(\tau')$ is defined.  We must
  show that $\hat{h}(\tau') \subty \hat{h}(\tau)$.  Since $\tau'
  \subty \tau$, by \refLem{l-monotonic} we have $L(\tau') \subseteq
  L(\tau)$.  It follows from the monotonicity of $\bigcup$, $L[-]$ and
  $\hat{h}[-]$ that $\bigcup L[\hat{h}[L(\tau')]] \subseteq \bigcup
  L[\hat{h}[L(\tau)]]$.  By \refLem{commutation}, we have that
  $L(\hat{h}(\tau')) \subseteq L(\hat{h}(\tau))$, but by
  \refLem{l-is-atomic-subtypes} this implies that $\hat{h}(\tau')
  \subty \hat{h}(\tau)$.
\end{proof}

\subsection{Proving algorithmic completeness}

The two key properties which ensure that occurrences of the
subsumption rule can be eliminated from derivations are
\emph{uniqueness of algorithmic types} and \emph{downward
  monotonicity} of the algorithmic judgments.  

Uniqueness, discussed already in proving decidability of the
algorithmic judgments (\refLem{decidability} and
\refLem{decidability-upd}), simply means that if the ``inputs'' to a
judgment are fixed, then there is at most one ``output'' type
derivable by algorithmic judgments; thus, the judgments define partial
functions.  Recall that  for fixed
$\bar{x},e,n,a,s$, we defined:
\begin{enumerate}
\item $f_n(\tau_1)$ as the unique $\tau_2$ such that $\tylab{\tau_1}{n}{\tau_2}$.
\item $g_e(\Gamma)$ as the unique $\tau$ such that
  $\wfalg{\Gamma}{e}{\tau}$ (if it exists).
\item $h_{\bar{x},e}(\Gamma,\tau_1)$ as the unique $\tau_2$ such that
  $\wfinalg{\Gamma}{x}{\tau_1}{e}{\tau_2}$ (if it exists).
\item $j_{a,s}(\Gamma,\tau_1)$ as the unique $\tau_2$ such that
  $\wfupdalg{\Gamma}{a}{\tau_1}{s}{\tau_2}$ (if it exists).
\item $k_{s}(\Gamma,\tau_1)$ as the unique $\tau_2$ such that
  $\wfiteralg{\Gamma}{\tau_1}{s}{\tau_2}$ (if it exists).
\end{enumerate}
Downward monotonicity of the type judgments corresponds precisely to
downward monotonicity of the above functions (where we use the
subtyping order on context arguments $\Gamma$ defined in
\refSec{query-types}.)  To prove downward monotonicity of the
type-directed $f,h,k$, we need to make use of the characterization of
downward monotonicity for partial homomorphic extensions established
in the last section.

\begin{proposition}[Downward Monotonicity]\labelProp{downward-monotonicity}
  ~\begin{enumerate}
    \item For every $n$, the function $f_n$ is downward monotonic.
    \item For every $e$ and $\bar{x}$, the functions $g_e$ and
      $h_{\bar{x},e}$ are downward monotonic, and
      $h_{\bar{x},e}(\Gamma,-)$ is the partial homomorphic extension
      of $g_e(\Gamma,\bar{x}{:}(-))$.
    \item For every $s$ and $a$, the functions $j_{a,s}$ and $k_s$ are
      downward monotonic, and $k_s(\Gamma,-)$ is the partial
      homomorphic extension of $j_{1,s}(\Gamma,-)$.
\end{enumerate}
\end{proposition}
\begin{proof}
  For part (1), we just need to show that $f_n$ is generated by the
  function
  \[\alpha \mapsto \left\{\begin{array}{ll}
      n[\tau] & \alpha = n[\tau]\\
      \emptyseq & \text{otherwise}
    \end{array}
  \right.
  \]
  which is obviously downward monotonic.
  
  For part (2), proof is by induction on the structure of $e$.  For
  each $e$, we first show downward monotonicity of $g_e$ by inspecting
  derivations.  We show a few representative examples:
\begin{itemize}
  \item Case (var): If the derivation is of the form
\[\infer{\wfalg{\Gamma}{\hat{x}}{\tau}}{\hat{x}{:}\tau \in \Gamma}
\]
then we have $\hat{x}:\tau' \in \Gamma'$ where $\tau' \subty \tau$, hence may derive:
\[\infer{\wfalg{\Gamma}{\hat{x}}{\tau'}}{\hat{x}{:}\tau' \in \Gamma'}
\]
\item Case ($\klet$): If the derivation is of the form
\[\infer{\wfalg{\Gamma}{\letin{x=e_1}{e_2}}{\tau_2}}
{\wfalg{\Gamma}{e_1}{\tau_1} & \wfalg{\Gamma,x{:}\tau_1}{e_2}{\tau_2}}
 \]
then by induction we have $\wfalg{\Gamma'}{e_1}{\tau_1'}$ for some $\tau_1' \subty \tau_1$ and since
$\Gamma' \subty \Gamma$, we have $\Gamma',x{:}\tau_1'\subty
\Gamma,x{:}\tau_1$, so also by induction
$\wfalg{\Gamma',x{:}\tau_1'}{e_2}{\tau_2'}$ for some $\tau_2' \subty \tau_2$.
To conclude, we derive
\[\infer{\wfalg{\Gamma'}{\letin{x=e_1}{e_2}}{\tau_2'}}
{\wfalg{\Gamma'}{e_1}{\tau_1'} & \wfalg{\Gamma',x{:}\tau_1'}{e_2}{\tau_2'}}
 \]
\item Case ($\kfor$): If the derivation is of the form
\[\infer{\wfalg{\Gamma}{\forreturn{x\in e_1}{e_2}}{\tau_2}}
{\wfalg{\Gamma}{e_1}{\tau_1} & \wfinalg{\Gamma}{x}{\tau_1}{e_2}{\tau_2}}
 \]
 then by induction we have $\wfalg{\Gamma'}{e_1}{\tau_1'}$ for some
 $\tau_1' \subty \tau_1$.  Using the downward monotonicity of
 $h_{\bar{x},e_2}$, we can obtain $\tau_2' \subty \tau_2$ such that $\wfinalg{\Gamma}{x}{\tau_1'}{e_2}{\tau_2'}$.  To conclude, we derive
\[\infer{\wfalg{\Gamma'}{\forreturn{x\in e_1}{e_2}}{\tau_2'}}
{\wfalg{\Gamma'}{e_1}{\tau_1'} & \wfinalg{\Gamma'}{x}{\tau_1'}{e_2}{\tau_2'}}
 \]
\end{itemize}

Showing that $h_{\bar{x},e}$ is downward monotonic is immediate once
we show that $h_{\bar{x},e}(\Gamma,-)$ is the partial homomorphic
extension of $g_e(\Gamma,\bar{x}{:}(-))$ for any $\Gamma$.  The latter
property can be proved by induction on the structure of derivations of
$\wfin{\Gamma}{x}{\tau_1}{e}{\tau_2}$.  The cases involving regular
expression constructs or variables are straightforward, and the base
case
\[
\infer{\wfin{\Gamma}{x}{\alpha}{e}{\tau}}{\wf{\Gamma,\bar{x}{:}\alpha}{e}{\tau}}\]
is also straightforward since $h_{\bar{x},e}(\Gamma,\tau) =
g_e(\Gamma,\bar{x}{:}\tau)$ by definition.

Similarly, for part (3), $j$ and $k$, the proof is by induction on
derivations.  The cases involving $j$ are straightforward; the case
involving $\nd_{\kiter}$ is similar to that for $\kfor$ above.  To
show $k_s(\Gamma,-)$ is the partial homomorphic extension of
$j_{1,s}(\Gamma,-)$ and hence that $k_s$ is downward monotonic, the
proof is by simultaneous induction on derivations, just as for
$g$ and $h$ above.
\end{proof}

By rewriting the above proposition in terms of judgments, we can conclude:
\begin{theorem}[Downward monotonicity]\labelThm{downward-monotonicity}
~
\begin{enumerate}
\item If $\tylab{\tau_1}{n}{\tau_2}$ and $\tau_1' \subty \tau_1$ then 
 $\tylab{\tau_1'}{n}{\tau_2'}$ for some $\tau_2' \subty \tau_2$
\item If $\wfalg{\Gamma}{e}{\tau}$ and $\Gamma' \subty \Gamma$ then
  $\wfalg{\Gamma'}{e}{\tau'}$ for some $\tau' \subty \tau$.
\item If $\wfinalg{\Gamma}{x}{\tau_1}{e}{\tau_2}$ and $\Gamma' \subty
  \Gamma$ and $\tau_1' \subty \tau_1$ then
  $\wfinalg{\Gamma'}{x}{\tau_1'}{e}{\tau_2'}$ for some $\tau_2' \subty
  \tau_2$.
\item If $\wfupdalg{\Gamma}{a}{\tau_1}{s}{\tau_2}$ and $\Gamma' \subty
  \Gamma$ and $\tau_1' \subty \tau_1$ then
  $\wfupdalg{\Gamma'}{a}{\tau_1'}{s}{\tau_2'}$ for some $\tau_2' \subty
  \tau_2$.
\item If $\wfiteralg{\Gamma}{\tau_1}{s}{\tau_2}$ and $\Gamma' \subty
  \Gamma$ and $\tau_1' \subty \tau_1$ then
  $\wfiteralg{\Gamma'}{\tau_1'}{s}{\tau_2'}$ for some $\tau_2' \subty
  \tau_2$.
\end{enumerate}
\end{theorem}

Finally, taking $\Gamma = \Gamma'$ and $\tau_1=\tau_1'$ in parts
2--5 above, we conclude:
\begin{theorem}[Algorithmic completeness]\labelThm{completeness-all}
  ~
  \begin{enumerate}
  \item If $\wf{\Gamma}{e}{\tau}$ then there exists $\tau' \subty
    \tau$ such that $\wfalg{\Gamma}{e}{\tau'}$.
  \item If $\wfin{\Gamma}{x}{\tau_1}{e}{\tau_2}$ then there exists $\tau_2' \subty
    \tau_2$ such that $\wfinalg{\Gamma}{x}{\tau_1}{e}{\tau_2'}$.
  \item If $\wfupd{\Gamma}{a}{\tau_1}{s}{\tau_2}$ then there exists
    $\tau_2' \subty \tau_2$ such that
    $\wfupdalg{\Gamma}{a}{\tau_1}{s}{\tau_2'}$
  \item If $\wfiter{\Gamma}{\tau_1}{s}{\tau_2}$ then there exists
    $\tau_2' \subty \tau_2$ such that
    $\wfiteralg{\Gamma}{\tau_1}{s}{\tau_2'}$
  \end{enumerate}
\end{theorem}

%
